\documentclass[a4paper,11pt,twoside,english]{article}
\usepackage{babel}
\usepackage{latexsym,amsfonts}
\usepackage{amsthm}
\usepackage{amsmath}
\usepackage{mathtools}
\usepackage{paralist}
\usepackage{booktabs}
\usepackage{multirow}
\usepackage{float}
\usepackage{scrextend}
\usepackage{siunitx}

\usepackage{lastpage}
\usepackage{fancyhdr}
\fancypagestyle{plain}{

	\fancyhead[L,C]{}
	\fancyhead[R]{\today}
	\fancyfoot[R,L]{}
	\fancyfoot[C]{\thepage \ of \pageref{LastPage}}
}

\title{On Buckingham's $\Pi$-Theorem}
\author{Jan-David Hardtke}
\date{}

\providecommand{\sm}{\setminus}
\providecommand{\ssq}{\subseteq}
\providecommand{\N}{\ensuremath{\mathbb{N}}}

\providecommand{\R}{\ensuremath{\mathbb{R}}}
\providecommand{\Ph}{\ensuremath{\mathcal{P}}}
\providecommand{\eps}{\ensuremath{\varepsilon}}
\providecommand{\Dim}[1]{\langle #1 \rangle}
\providecommand{\one}{\ensuremath{\underline{1}}}
\providecommand{\Matrix}[2][1]{\left(\begin{array}{*{#1}{c}} #2 \end{array}\right)}

\providecommand{\keywords}[1]{
{\let\thefootnote=\relax
\footnote{{\em Keywords}: #1}}
\addtocounter{footnote}{-1}
}

\providecommand{\AMS}[1]{
{\let\thefootnote=\relax
\footnote{{\em AMS Subject Classification} (2010): #1}}
\addtocounter{footnote}{-1}
}

\providecommand{\address}{
{\sc \noindent Department of Mathematics \\
Universit\"at Leipzig\\
Augustusplatz 10, 04109 Leipzig \\
Germany \\}
}

\DeclarePairedDelimiter{\set}{\lbrace}{\rbrace}
\DeclarePairedDelimiter{\paren}{\lparen}{\rparen}
\DeclarePairedDelimiter{\abs}{\lvert}{\rvert}
\DeclarePairedDelimiter{\norm}{\lVert}{\rVert}

\theoremstyle{definition}
\newtheorem{definition}{Definition}[section]
\newtheorem*{definition*}{Definition}
\theoremstyle{plain}
\newtheorem{lemma}[definition]{Lemma}
\newtheorem*{lemma*}{Lemma}
\newtheorem{theorem}[definition]{Theorem}
\newtheorem*{theorem*}{Theorem}

\newenvironment{Abstract}{\centering\begin{minipage}{0.8\textwidth} \noindent \small {\sc Abstract.}}{\end{minipage}\par}

\usepackage{color}
\definecolor{darkgreen}{rgb}{0,0.5,0}

\usepackage{footnote}

\numberwithin{equation}{section}
\newtagform{colored}[\color{blue}]{\color{blue}(}{\color{blue})}
\usetagform{colored}

\usepackage[colorlinks,linkcolor=blue,citecolor=red,urlcolor=darkgreen]{hyperref}
\providecommand{\email}{{\it E-mail address:} \href{mailto:hardtke@math.uni-leipzig.de}{\tt hardtke@math.uni-leipzig.de}}

\usepackage{amsrefs}

\begin{document}
	
\maketitle	
	
\begin{Abstract}
Roughly speaking, Buckingham's $\Pi$-Theorem provides a method to ``guess'' the structure of 
physical formulas simply by studying the dimensions (the physical units) of the involved quantities. Here we will prove a quantitative version of Buckingham's Theorem, which is 
``purely mathematical'' in the sense that it does make any explicit reference to physical 
units.
\end{Abstract}
\keywords{physical units; Buckingham's theorem; Moore-Penrose pseudoinverse}
\AMS{46N99 15A99 15A09}

\section{Introduction: Physical Units}\label{sec:intro}
Let us begin with a brief discussion on physical units. Every physical quantity consists of a numerical value and a unit of measure. The most commonly used  system of units is the SI (syst\`eme internationale (d'unit\'es)). It uses the following seven base units.\\
{	
	\renewcommand{\arraystretch}{1.3}
\begin{center}
{\bf Table 1: SI base units}\\
\ \\
\begin{tabular}{|c|c|}
\hline
{\bf Quantity} \ \ & {\bf Unit} \\ \hline
length \ \ & meter (\si{m}) \\ \hline
mass \ \ & kilogram (\si{kg}) \\ \hline
time \ \ & second (\si{s}) \\ \hline
electric current \ \ & ampere (\si{A}) \\ \hline
temperature \ \ & kelvin (\si{K}) \\ \hline
amount of substance \ \ & mole (\si{mol}) \\ \hline
luminous intensity \ \ & candela (\si{cd}) \\ \hline	
\end{tabular}
\end{center}
	}
\ \\
\indent All other units are of the form $\si{m^{\alpha_1}.kg^{\alpha_2}.s^{\alpha_3}.A^{\alpha_4} .K^{\alpha_5}.mol^{\alpha_6}.cd^{\alpha_7}}$ with rational exponents $\alpha_1,\dots,\alpha_7$.
A table of some of the most important derived physical quantities and their SI units 
is included in the Appendix (Table 2).\par
A quantity's unit (or more properly the corresponding tuple of exponents) is called its dimension. If all the exponents are zero, the quantity is called dimensionless (i.\,e. it is just a number without a unit).\par 
When two physical quantities, say $p$ and $q$ with dimensions $(\alpha_1,\dots,\alpha_7)$ and $(\beta_1,\dots,\beta_7)$, are multiplied, their units are multiplied accordingly, i.\,e. the dimension of $pq$ is $(\alpha_1+\beta_1,\dots,\alpha_7+\beta_7)$.\par
The sum $p+q$ is only defined if $p$ and $q$ have the same dimension (which is then also the dimension of $p+q$).\par
In addition, we may also use scaling factors together with the units. When taking powers of units these factors have to be scaled accordingly, for instance $1\,\si{cm}=10^{-2}\,\si{m}$ 
and $1\,\si{cm^3}=(10^{-2})^3\,\si{m^3}=10^{-6}\,\si{m^3}$.\par 
Table 3 in the Appendix shows the standard scaling factors and their prefixes used in the SI. Other scaling factors may occur when converting non-SI to SI units. A list of some non-SI units which are still frequently used (at least in some specific areas) can be found in Table 4 in the Appendix.\par 
When confronted with a physical formula, usually the first thing one does is to check whether
it is consistent, i.\,e. whether the physical dimensions on both sides of the equation agree (otherwise it cannot be true).\par 
Of course, consistency of the dimensions is not sufficient to ensure that a given formula holds true. Nonetheless, in many situations a closer look at the dimensions of the involved quantities may give us a hint as to how a sought formula might look like. Let us consider the following classical example: we have a mathematical pendulum, that is, a small mass $m$ attached to a string of length $l$ suspended from the ceiling, so that the mass $m$ can swing freely back and forth under the influence of gravity $g$.\par 
For small enough elongations, the mass $m$ performs a periodic motion with period $T$. We wish to express $T$ as a function of the quantities $m$, $l$ and $g$. They have the SI units $\si{kg}$, $\si{m}$ and $\si{m/s^2}$, while $T$ has the unit $\si{s}$. It seems reasonable to assume that the unit of $T$ should come out as some combination of the units of $m$, $l$ and $g$. Thus we try to find rational numbers $y_1,y_2,y_3$ such that 
\begin{equation*}
\si{s}=\si{kg^{y_1}.m^{y_2}.(m.s^{-2})^{y_3}}=\si{kg^{y_1}.m^{y_2+y_3}.s^{-2y_3}}.
\end{equation*} 
\indent Comparing the exponents on both sides leads to $y_1=0$, $y_2+y_3=0$ and $-2y_3=1$, hence $y_3=-1/2$ and $y_2=1/2$.\par
Thus we may conjecture that the formula for $T$ is of the form $m^{y_1}l^{y_2}g^{y_3}=\sqrt{l/g}$, possibly times some dimensionless factor $C$:
\begin{equation*}
T=C\sqrt{\frac{l}{g}}.
\end{equation*}
\indent Surprisingly, for $C=2\pi$ this gives indeed the correct formula. The general principle behind this reasoning is what is known as Buckingham's $\Pi$-Theorem. To formulate it, we consider not only the SI with its seven base units but the more general case of $m$ base units $U_1,\dots,U_m$. The unit of a physical quantity $p$ is then of the form $U_1^{\alpha_1}\dots U_m^{\alpha_m}$, where we may even allow the exponents to be irrational. The tuple
$(\alpha_1,\dots,\alpha_m)$ is again called the dimension of $p$.\par 
Now suppose we have a physical quantity $q$ with dimension $\beta=(\beta_1,\dots,\beta_m)$. We assume that $q$ depends on the physical quantities $p_1,\dots,p_n$, i.\,e. we assume that $q=F(p_1,\dots,p_n)$ for a suitable function $F$ which
we would like determine. We denote the dimension of $p_j$ by $\alpha_j=(\alpha_{1j},\dots,\alpha_{mj})$ and consider the $m\times n$-matrix $A$ whose $j$-th column is $\alpha_j^T$. The unit of the output $q$ should be a combination of the units of the inputs $p_1,\dots,p_n$, so we assume that there are $y_1,\dots,y_n\in \R$ such that
\begin{equation*}
U_1^{\beta_1}\dots U_m^{\beta_m}=(U_1^{\alpha_{11}}\dots U_m^{\alpha_{m1}})^{y_1}\dots (U_1^{\alpha_{1n}}\dots U_m^{\alpha_{mn}})^{y_n}, 
\end{equation*}
which is equivalent to $Ay^T=\beta^T$, where $y:=(y_1,\dots,y_n)$.\par 
Moreover, the function $F$ should also have the correct scaling behavior. If we scale the units of $p_1,\dots,p_n$, then the unit of the output $q=F(p_1,\dots,p_n)$ should be scaled accordingly, i.\,e. for all $c_1,\dots,c_m>0$ and all $v_1,\dots,v_n$ we have 
\begin{align*}
&F(v_1(c_1U_1)^{\alpha_{11}}\dots (c_mU_m)^{\alpha_{m1}},\dots,v_n(c_1U_1)^{\alpha_{1n}}\dots (c_mU_m)^{\alpha_{mn}})\\
&=F(v_1c_1^{\alpha_{11}}\dots c_m^{\alpha_{m1}}U_1^{\alpha_{11}}\dots U_m^{\alpha_{m1}},\dots,
v_nc_1^{\alpha_{1n}}\dots c_m^{\alpha_{mn}}U_1^{\alpha_{1n}}\dots U_m^{\alpha_{mn}})\\
&=F(v_1U_1^{\alpha_{11}}\dots U_m^{\alpha_{m1}},\dots,v_nU_1^{\alpha_{1n}}\dots U_m^{\alpha_{mn}})c_1^{\beta_1}\dots c_m^{\beta_m}.
\end{align*}
\indent Now we distinguish two cases. First we assume that the rank of $A$ is equal to $n$. Then Buckingham's Theorem states that $F$ is of the form 
\begin{equation*}
F(p_1,\dots,p_n)=Cp_1^{y_1}\dots p_n^{y_n}
\end{equation*}
for a dimensionless constant $C$. This is what we got in the previous example for the period $T$
(here $\alpha_1=(0,1,0)$, $\alpha_2=(1,0,0)$, $\alpha_3=(1,0,-2)$ and $\beta=(0,0,1)$).\par 
The second case, $\mathrm{rank}(A)<n$, is a bit more complicated. Buckingham's $\Pi$-Theorem states that there are $k=n-\mathrm{rank}(A)$ dimensionless quantities $\pi_1,\dots,\pi_k$ and a suitable function $G$ such that 
\begin{equation*}
F(p_1,\dots,p_n)=G(\pi_1,\dots,\pi_k)p_1^{y_1}\dots p_n^{y_n}.
\end{equation*}
\indent The $\pi_s$ can be determined as follows: choose a basis $(x_1^T,\dots,x_k^T)$ of the kernel of $A$ and denote the $j$-th coordinate of $x_s$ by $x_{sj}$. Then $\pi_s=p_1^{x_{s1}}\dots p_n^{x_{sn}}$ for $s=1,\dots,k$ (this is not unique, a different choice
for the $x_s$ may lead to a different set of quantities $\pi_1,\dots,\pi_k$ and hence also to a different $G$).\par
As an example let us consider a  classical Atwood machine consisting of two masses $m_1$ and $m_2$ connected by an inextensible, massless string over a massless, frictionless pulley. In the beginning, both masses are assumed to be at rest at the same height $h$. Then the larger mass drops to the ground while the smaller mass is pulled up to the height $2h$. We want to find the velocity $v$ of the larger mass at the moment it hits the ground. This velocity should be a function of $m_1$, $m_2$, $h$ and the gravity $g$, so $v=F(m_1,m_2,h,g)$.\par 
The matrix $A$ whose columns are given by the dimensions of $m_1$, $m_2$, $h$ and $g$ is 
\begin{equation*}
A=\Matrix[4]{0 & 0 & 1 & 1\\ 1 & 1 & 0 & 0 \\ 0 & 0 & 0 & -2},
\end{equation*}
while the dimension of $v$ is $\beta=(1,0,-1)$.\par
Obviously $\mathrm{rank}(A)=3$, the kernel of $A$ is spanned by $x^T=(1,-1,0,0)^T$ and for $y=(0,0,1/2,1/2)$ we have $Ay^T=\beta^T$.\par
So if we assume that the function $F$ has the correct scaling behavior as described above, then 
Buckingham's Theorem implies that there is a function $G$ of one variable such that 
\begin{equation*}
v=F(m_1,m_2,h,g)=G(\pi_1)m_1^{y_1}m_2^{y_2}h^{y_3}g^{y_4}=G(\pi_1)\sqrt{gh},
\end{equation*}
where the dimensionless quantity $\pi_1$ is given by $\pi_1=m_1^{x_1}m_2^{x_2}h^{x_3}g^{x_4}=m_1/m_2$.\par 
In summary we have $v=G(m_1/m_2)\sqrt{gh}$ for an unknown function $G$. While this formula is of course not a complete solution to our problem, it still gives us some valuable information:
we see how $v$ depends on $g$ and $h$ if $m_1$ and $m_2$ are fixed. Moreover, $v$ does not directly depend on the masses $m_1$ and $m_2$ themselves, only on their ratio $m_1/m_2$.\par
Using for instance the law of conversation of energy, one can derive the exact formula
\begin{equation*}
v=\sqrt{2gh\frac{|m_1-m_2|}{m_1+m_2}}=G(m_1/m_2)\sqrt{gh},
\end{equation*}
where
\begin{equation*}
G(z)=\sqrt{2\frac{|z-1|}{z+1}}.
\end{equation*}
\par 
Alternatively, one could also consider the forces acting on $m_1$ and $m_2$ and employ Newton's second law of motion to get the same result.\par
Now some historical comments are in order. Buckingham's Theorem is named for Edgar Buckingham, who described it in his work \cite{buckingham} in 1914. But essentially the same result appeared already in 1892 in an article by A. Vaschy \cite{vaschy} and also in 1911 in works 
of A. Federman \cite{federman} and D. Riabouchinsky \cite{riabouchinsky}. It was Buckingham who introduced the symbol $\pi$ for the dimensionless quantities, which let to the name $\Pi$-Theorem. For some modern works on the subject, see for instance \cites{boyling,curtis,hanche-olsen}.\par 
So far our discussion of Buckingham's Theorem has not been quite rigorous from a mathematical point of view, largely due to the fact that the physical units themselves are not properly defined mathematical objects. Now we will define a purely mathematical formalism to model the idea of physical quantities without explicit reference to physical units (of course this is known material, it is included here only for the reader's convenience).

\section{Mathematical Formalism}\label{sec:math}
First we introduce some notation. We set $\R_+:=(0,\infty)$ and for a fixed natural number $m$, we denote by $\R^m$ the space of all column vectors of length $m$ with real entries, while $\R_m$ denotes the corresponding space of row vectors. The symbols $\R_+^m$ resp. $\R_m^+$ denote the set of all elements of $\R^m$ resp. $\R_m$ with positive entries.\par
For $c=(c_1,\dots,c_m)\in \R_m^+$ and $\alpha=(\alpha_1,\dots,\alpha_m)\in \R_m$ we put $c^{\alpha}:=\prod_{i=1}^mc_i^{\alpha_i}$.\par
Now we can formulate the following definition.
\begin{definition}\label{def:phys}
We define a relation $\sim$ on $\R\times \R_m^+\times \R_m$ in the following way:
\begin{equation*}
(x,c,\alpha)\sim (y,d,\beta) \ \ :\Leftrightarrow \ \ \alpha=\beta \ \text{and}\ xc^{\alpha}=yd^{\alpha},
\end{equation*}
where $x,y\in \R$, $c, d\in \R_m^+$ and $\alpha, \beta\in \R_m$.
\end{definition}

Obviously $\sim$ is an equivalence relation. We write $[x,c,\alpha]$ for the equivalence class of $(x,c,\alpha)$. 
The set of all equivalence classes is denoted by 
\begin{equation*}
\Ph:=\set*{[x,c,\alpha]:(x,c,\alpha)\in \R\times \R_m^+\times \R_m}
\end{equation*}
and for each $\alpha\in \R_m$ we set 
\begin{equation*}
\Ph_{\alpha}:=\set*{[x,c,\alpha]:x\in \R, c\in \R_m^+}.
\end{equation*}
Then we have
\begin{equation*}
\Ph=\biguplus_{\alpha\in \R_m}\Ph_{\alpha}.
\end{equation*}

Now suppose that $p\in \Ph_{\alpha}$. Then we call $\alpha$ the {\it dimension} of $p$, denoted by $\Dim{p}=\alpha$ (if $\Dim{p}=0$, then $p$ is called {\it dimensionless}). \par 
For every $c\in \R_m^+$ there exists a unique $x\in \R$ such that $p=[x,c,\alpha]$, which we denote by  $[p]_c=x$. Also, we put $[p]:=[p]_{\one}$, where $\one:=(1,\dots,1)$.\par
Obviously, if $c=(c_1,\dots,c_m),d=(d_1,\dots,d_m)\in \R_m^+$, then 
\begin{equation*}
[p]_d=[p]_c(c/d)^{\alpha},
\end{equation*}
where $c/d:=(c_1/d_1,\dots,c_m/d_m)$.\par 
The elements of $\Ph_{\alpha}$ can indeed be interpreted as physical quantities of dimension $\alpha$ with respect to some fixed set of $m$ base units. $[p]_c$ can be understood as the numerical value of the quantity $p$ when the base units are scaled with the factors $c_1,\dots,c_m$.\par 
Next we define a multiplication on $\Ph$ which models the multipilcation of physical quantities.
\begin{definition}\label{def:mult}
For all $p_1,p_2\in \Ph$ we set 
\begin{equation*}
p_1p_2:=[[p_1][p_2],\one,\Dim{p_1}+\Dim{p_2}].
\end{equation*}
\end{definition}

The following assertions are easily proved.
\begin{lemma}\label{lemma:mult}
For all $p_1,p_2\in \Ph$ we have $\Dim{p_1p_2}=\Dim{p_1}+\Dim{p_2}$ and 
$[p_1p_2]_c=[p_1]_c[p_2]_c$ for every $c\in \R_m^+$.\par
Moreover, $(\Ph,\cdot)$ is a commutative monoid with neutral element $[1,\one,0]$
and $p\in \Ph$ is invertible if and only if $[p]\neq 0$. In that case,
$\Dim{p^{-1}}=-\Dim{p}$ and $[p^{-1}]_c=[p]_c^{-1}$ for every $c\in \R_m^+$.
\end{lemma}

We can also define an operation $+$ on each set $\Ph_{\alpha}$, mimicking again the summation of
physical quantities.
\begin{definition}\label{def:sum}
Let $\alpha\in \R_m$. For $p_1,p_2\in \Ph_{\alpha}$ we define
\begin{equation*}
p_1+p_2:=[[p_1]+[p_2],\one,\alpha].
\end{equation*}
\end{definition}

Then the following lemma holds (obviously).
\begin{lemma}\label{lemma:sum}
Let $\alpha\in \R_m$. Then $(\Ph_{\alpha},+)$ is an abelian group and for each $c\in \R_m^+$
and all $p_1,p_2\in \Ph_{\alpha}$ we have $[p_1+p_2]_c=[p_1]_c+[p_2]_c$.
\end{lemma}

Next we also define positivity in $\Ph$.
\begin{definition}\label{def:ordering}
We say that $p\in \Ph$ is positive ($p>0$) if $[p]>0$. For every $\alpha\in \R^m$ we put 
\begin{equation*}
\Ph_{\alpha}^+:=\set*{p\in \Ph_{\alpha}:p>0}.
\end{equation*}
\end{definition}

Obviously, $p>0$ if and only if $[p]_c>0$ for some $c\in \R_+^m$ if and only if $[p]_c>0$ for
every $c\in \R_+^m$.\par  
Now we define the class of functions which have the correct scaling behavior for Buckingham's Theorem. Since we want to prove a quantitative version of this theorem, we will also consider 
functions which have ``almost'' the correct scaling behavior (up to $\eps$).
\begin{definition}\label{def:functions}
Given $n,m\in \N$, $\eps\geq 0$, $\beta\in \R_m$ and a real $m\times n$ matrix $A$ with columns $\alpha_1^T,\dots,\alpha_n^T$, we denote by $S(A,\beta,\eps)$ the set of all functions 
$F:\R_n^+ \rightarrow \R$ which satisfy
\begin{equation*}
|F(v_1c^{\alpha_1},\dots,v_nc^{\alpha_n})-F(v_1,\dots,v_n)c^{\beta}|\leq\eps|F(v_1,\dots,v_n)|c^{\beta}
\end{equation*}
for all $v_1,\dots,v_n\in \R_+$ and every $c\in \R_m^+$. For $S(A,\beta,0)$ we simply write $S(A,\beta)$.
\end{definition}

The elements of $S(A,\beta)$ (the usual class of functions considered in Buckingham's Theorem)
are exactly those for which the scaling of physical units is consistent. Precisely, this means the following: if $F:\R_n^+ \rightarrow \R$ is given, define a map 
$\widehat{F}:\Ph_{\alpha_1}^+\times\dots\times\Ph_{\alpha_n}^+ \rightarrow \Ph_{\beta}$ by 
\begin{equation*}
\widehat{F}(p_1,\dots,p_n):=[F([p_1],\dots,[p_n]),\one,\beta] \ \ \ \forall (p_1,\dots,p_n)\in \Ph_{\alpha_1}^+\times\dots\times\Ph_{\alpha_n}^+
\end{equation*}
and then it is easy to see that the following assertions are equivalent:
\begin{enumerate}[(i)]
\item $F\in S(A,\beta)$
\item For all $(p_1,\dots,p_n)\in \Ph_{\alpha_1}^+\times\dots\times\Ph_{\alpha_n}^+$ and all 
$c\in \R_m^+$ we have 
\begin{equation*}
\widehat{F}(p_1,\dots,p_n)=[F([p_1]_c,\dots,[p_n]_c),c,\beta].
\end{equation*}
\end{enumerate}
\par 
In the next section we shall need the following lemma (it is of course well-known, but we include a proof here for the reader's convenience).
\begin{lemma}\label{lemma:buckingham}
Let $A$ be a real $m\times n$ matrix, $\beta\in \R_m$ and $y\in \R_n$ such that $Ay^T=\beta^T$.	
Suppose that $G\in S(A,0)$ and put $F(v):=G(v)v^y$ for all $v\in \R_n^+$. Then $F\in S(A,\beta)$.
\end{lemma}

\begin{proof}
Let $y=(y_1,\dots,y_n)$, $\beta=(\beta_1,\dots,\beta_m)$ and $A=(\alpha_{ij})_{i,j=1}^{m,n}$. Put $\alpha_j:=(\alpha_{1j},\dots,\alpha_{mj})$ for $j=1,\dots,n$.\par 
For all $v=(v_1,\dots.v_n)\in \R_n^+$ and $c=(c_1,\dots.c_m)\in \R_m^+$ we have
\begin{align*}
&F(v_1c^{\alpha_1},\dots,v_nc^{\alpha_n})=G(v_1c^{\alpha_1},\dots,v_nc^{\alpha_n})\prod_{j=1}^n(v_jc^{\alpha_j})^{y_j}\\
&=G(v_1,\dots,v_n)\paren*{\prod_{j=1}^nv_j^{y_j}}\paren*{\prod_{j=1}^nc^{y_j\alpha_j}}=F(v_1,\dots,v_n)\prod_{i=1}^m\prod_{j=1}^nc_i^{\alpha_{ij}y_j}\\
&=F(v_1,\dots,v_n)\prod_{i=1}^mc_i^{\sum_{j=1}^n\alpha_{ij}y_j}=F(v_1,\dots,v_n)\prod_{i=1}^mc_i^{\beta_i}=F(v_1,\dots,v_n)c^{\beta}.
\end{align*}\par 
Thus $F\in S(A,\beta)$.
\end{proof}

We will also need generalized inverses. If $B=(b_{ij})_{i,j=1}^{m,n}$ is a real $m\times n$-matrix, then there exists exactly one real $n\times m$ matrix $B^{\dagger}$ such that the following conditions are satisfied:
\begin{enumerate}[(i)]
\item $BB^{\dagger}B=B$
\item $B^{\dagger}BB^{\dagger}=B^{\dagger}$
\item $BB^{\dagger}$ and $B^{\dagger}B$ are symmetric.	
\end{enumerate}\par
$B^{\dagger}$ is called the Moore-Penrose pseudoinverse of $B$.\par 
If $\mathrm{rank}(B)=m$, then $BB^T$ is invertible and $B^{\dagger}=B^T(BB^T)^{-1}$ (hence $BB^{\dagger}=I_m$, where $I_m$ denotes the $m\times m$ identity matrix). Likewise, if $\mathrm{rank}(B)=n$, then $B^TB$ is invertible and $B^{\dagger}=(B^TB)^{-1}B^T$
(hence $B^{\dagger}B=I_n$). For this and more information on pseudoinverses, see for instance \cite{dym}*{Chapter 11}.\par
Furthermore, we denote by $\norm{\cdot}_{\infty}$ the maximum-norm on $\R^n$ resp. $\R^m$ and by $\norm{B}$ the corresponding matrix (operator) norm, which---as is well-known---can be expressed as 
\begin{equation*}
\norm{B}=\max_{i=1,\dots,m}\sum_{j=1}^n|b_{ij}|.
\end{equation*}

Finally, let us introduce a last bit of notation. For $v=(v_1,\dots,v_n)\in \R_n^+$ we define $\log(v):=(\log(v_1),\dots,\log(v_n))$ (where $\log$ denotes the natural logarithm) and for $w=(w_1,\dots,w_n)^T\in \R^n$ we define $\exp(w):=(e^{w_1},\dots,e^{w_n})$.\par

Now we are ready to formulate and prove quantitative versions of Buckingham's Theorem, where we replace the condition $F\in S(A,\beta)$ by the approximate condition $F\in S(A,\beta,\eps)$ and the condition $Ay^T=\beta^T$ by the approximate condition $\norm{Ay^T-\beta^T}_{\infty}\leq \delta$.\par 
Note that the formulations and proofs are ``purely mathematical'', avoiding any explicit reference to physical units (the paper \cite{curtis} also contains an abstract, but non-quantitative, version of Buckingham's Theorem without reference to physical units; 
the proofs we give below for Theorems \ref{thm:buckingham1} and \ref{thm:buckingham2}
were originally inspired by the notes of H. Hanche-Olsen \cite{hanche-olsen} on the classical Buckingham Theorem).

\section{Quantitative Version of Buckingham's Theorem}\label{sec:buckingham}
Again we have to distinguish the cases $\mathrm{rank}(A)=n$ and $\mathrm{rank}(A)<n$.
We start with the easier case of full column-rank.

\begin{theorem}[Quantitative version of Buckingham's Theorem, part 1]\label{thm:buckingham1}
\ \\	
Let $A$ be a real $m\times n$-matrix with $\mathrm{rank}(A)=n$ and let $\beta\in \R_m$, $\eps\geq 0$ and $F\in S(A,\beta,\eps)$. Furthermore, let $y\in \R_n$ and $\delta\geq 0$ with $\norm{Ay^T-\beta^T}_{\infty}\leq\delta$.\par
Put $C:=F(1,\dots,1)$ and $D:=(A^T)^{\dagger}$ and fix a number $K>1$.\par 
If $v=(v_1,\dots,v_n)\in \R_n^+$ is such that $K^{-1}\leq v_i\leq K$ for all $i=1,\dots,n$,
then we have
\begin{equation*}
|F(v)-Cv^y|\leq|F(v)|((1+\eps)K^{m\delta\norm{D}}-1).
\end{equation*} 
\end{theorem}

Roughly speaking, the theorem states that for small enough $\eps$ and $\delta$ the quotient
$Cv^y/F(v)$ is close to one.\par 
As a particular case, we get that if $Ay^T=\beta^T$ and $F\in S(A,\beta)$, then $F(v)=Cv^y$ holds for every $v\in \R_n^⁺$, which is the content of the usual Buckingham-Theorem for $\mathrm{rank}(A)=n$.

\begin{proof}[Proof of Theorem 3.1]
Write $y=(y_1,\dots,y_n)$, $\beta=(\beta_1,\dots,\beta_m)$ and $A=(\alpha_{ij})_{i,j=1}^{m,n}$. Put $\alpha_j:=(\alpha_{1j},\dots,\alpha_{mj})$ for $j=1,\dots,n$ and $z_i:=(\alpha_{i1},\dots,\alpha_{in})$ for $i=1,\dots,m$.\par 
Let 
\begin{equation*}
(\lambda_1,\dots,\lambda_m)^T:=-D\log(v)^T.
\end{equation*}\par 
Since $A^TD=I_n$, we have
\begin{equation*}
\sum_{i=1}^m\lambda_iz_i=-\log(v).
\end{equation*}\par 
Put $c_i:=e^{\lambda_i}$ for $i=1,\dots,m$ and $c:=(c_1,\dots,c_m)$. It follows that
\begin{equation*}
\log(c^{\alpha_j})=\log\paren*{\prod_{i=1}^mc_i^{\alpha_{ij}}}=\sum_{i=1}^m\alpha_{ij}\lambda_i=-\log(v_j)
\end{equation*}
and hence $c^{\alpha_j}=1/v_j$ for $j=1,\dots,n$.\par 
Since $F\in S(A,\beta,\eps)$ this implies
\begin{equation}\label{eq:buckingham1}
|C-F(v)c^{\beta}|\leq\eps|F(v)|c^{\beta}.
\end{equation}\par 
Now let $\gamma=(\gamma_1,\dots,\gamma_m):=yA^T=(Ay^T)^T$. It follows from Lemma \ref{lemma:buckingham} that $v^yc^{\gamma}=\prod_{j=1}^n(v_jc^{\alpha_j})^{y_j}=1$.\par 
Together with \eqref{eq:buckingham1} this implies 
\begin{align}\label{eq:buckingham2}
&|F(v)-Cv^y|\leq|F(v)-F(v)c^{\beta}v^y|+|F(v)c^{\beta}v^y-Cv^y| \nonumber\\
&\leq|F(v)|(|c^{\beta-\gamma}-1|+\eps c^{\beta-\gamma}).
\end{align}\par 
Our assumption on $v$ implies $\norm{\log(v)^T}_{\infty}\leq \log(K)$ and thus we have
$\norm{(\lambda_1,\dots,\lambda_m)^T}_{\infty}=\norm{D\log(v)^T}_{\infty}\leq \norm{D}\log(K)$.\par 
But then $c_i=e^{\lambda_i}\in [K^{-\norm{D}},K^{\norm{D}}]$ holds for $i=1,\dots,m$.\par 
Since $|\beta_i-\gamma_i|\leq \norm{\beta^T-\gamma^T}_{\infty}\leq \delta$ it follows that $c_i^{\beta_i-\gamma_i}\in [K^{-\delta\norm{D}},K^{\delta\norm{D}}]$ for $i=1,\dots,m$ and thus $c^{\beta-\gamma}\in [K^{-m\delta\norm{D}},K^{m\delta\norm{D}}]$.\par
Hence 
\begin{equation*}
|c^{\beta-\gamma}-1|\leq\max\set*{K^{m\delta\norm{D}}-1,1-K^{-m\delta\norm{D}}}
=K^{m\delta\norm{D}}-1
\end{equation*}
(here we have used the inequality $x+x^{-1}\geq 2$ for $x>0$).\par
Combining this with \eqref{eq:buckingham2} we obtain 
\begin{equation*}
|F(v)-Cv^y|\leq|F(v)|((1+\eps)K^{m\delta\norm{D}}-1).
\end{equation*} 
\end{proof}

Now we turn to the second case $\mathrm{rank}(A)<n$.
 
\begin{theorem}[Quantitative version of Buckingham's Theorem, part 2]\label{thm:buckingham2}
\ \\	
Let $A$ be a real $m\times n$-matrix with $r:=\mathrm{rank}(A)<n$ and let $k:=n-r$.\par 
Let $\beta\in \R_m$, $\eps\geq 0$ and $F\in S(A,\beta,\eps)$. Also, let $y\in \R_n$ and $\delta\geq 0$ such that $\norm{Ay^T-\beta^T}_{\infty}\leq\delta$.\par 
Put $D:=(A^T)^{\dagger}$ and fix a basis $(x_1^T,\dots,x_k^T)$ of the kernel $\mathrm{ker}(A)$ of $A$. Denote by $x_{sj}$ the $j$-th coordinate of $x_s$ for $s=1,\dots,k$ and $j=1,\dots,n$ and  put $X:=(x_{sj})_{s,j=1}^{k,n}$.\par 
Define functions $\pi_s:\R_n^+ \rightarrow \R_+$ by $\pi_s(u):=u^{x_s}$ for all $u\in \R_n^+$ and $s=1,\dots,k$.\par 
Furthermore, define 
\begin{equation*}
\psi(w):=\exp(X^{\dagger}\log(w)^T) \ \  \text{and}\ \ \ G(w):=F(\psi(w))/\psi(w)^y \ \ \ \forall w\in \R_k^+.
\end{equation*}\par
Let $M:=\max\set*{|x_{sj}|:s=1,\dots,k,\,j=1,\dots,n}$ and let $K>1$.\par 
Then the following assertions hold:
\begin{enumerate}[\upshape(a)]
\item $\pi_1,\dots,\pi_k\in S(A,0)$
\item If $v=(v_1,\dots,v_n)\in \R_n^+$ such that $K^{-1}\leq v_i\leq K$ for all $i=1,\dots,n$,
then we have 
\begin{equation*}
|F(v)-G(\pi_1(v),\dots,\pi_k(v))v^y|\leq |F(v)|((1+\eps)K^{m\delta\norm{D}(nM\norm{X^{\dagger}}+1)}-1).
\end{equation*} 	
\end{enumerate}
\end{theorem}

Statement (a) expresses the fact that the quantities associated to $\pi_1,\dots,\pi_k$ are dimensionless, while statement (b) can roughly be interpreted as ``the quotient $G(\pi_1(v),\dots,\pi_k(v))v^y/F(v)$ is close to one if $\eps$ and $\delta$ are small.''\par
In particular, if $Ay^T=\beta^T$ and $F\in S(A,\beta)$, then we obtain $F(v)=G(\pi_1(v),\dots,\pi_k(v))v^y$ for all $v\in \R_n^+$, which is the usual Buckingham-Theorem.

\begin{proof}[Proof of Theorem 3.2]
Let $y=(y_1,\dots,y_n)$, $\beta=(\beta_1,\dots,\beta_m)$ and $A=(\alpha_{ij})_{i,j=1}^{m,n}$. Put $\alpha_j:=(\alpha_{1j},\dots,\alpha_{mj})$ for $j=1,\dots,n$.\par 
(a) Since $Ax_s^T=0$ it follows from Lemma \ref{lemma:buckingham} that $\pi_s\in S(A,0)$ for
$s=1,\dots,k$.\par
(b) First we make the following observation: if $w=(w_1,\dots,w_k)\in \R_k^+$ and $u=(u_1,\dots,u_n)\in \R_n^+$, then we have 
\begin{equation}\label{eq:buckingham3}
(\pi_1(u),\dots,\pi_k(u))=w \ \ \Leftrightarrow \ \ X\log(u)^T=\log(w)^T.
\end{equation}\par 
This can be seen as follows:
\begin{align*}
&(\pi_1(u),\dots,\pi_k(u))=w \ \ \ \Leftrightarrow \ \ \ w_s=\prod_{j=1}^nu_j^{x_{sj}} \ \ \forall s=1,\dots,k\\
&\Leftrightarrow \ \ \ \log(w_s)=\sum_{j=1}^nx_{sj}\log(u_j) \ \ \forall s=1,\dots,k \ \ \ \Leftrightarrow \ \ \ X\log(u)^T=\log(w)^T.
\end{align*}\par
Now we fix $v=(v_1,\dots,v_n)\in \R_n^+$ such that $K^{-1}\leq v_i\leq K$ for all $i=1,\dots,n$ and put $w=(w_1,\dots,w_k):=(\pi_1(v),\dots,\pi_k(v))$ and $u=(u_1,\dots,u_n):=\psi(w)$.\par 
Then we have $X\log(u)^T=XX^{\dagger}\log(w)^T=\log(w)^T$ (note that $\mathrm{rank}(X)=k$, thus 
$XX^{\dagger}=I_k$). Hence by \eqref{eq:buckingham3} we have $(\pi_1(u),\dots,\pi_k(u))=w=(\pi_1(v),\dots,\pi_k(v))$.\par 
Another application of \eqref{eq:buckingham3} gives $\log(u/v)^T=\log(u)^T-\log(v)^T\in \mathrm{ker}(X)$.\par 
Denote by $z_i$ the $i$-th row of $A$. Then, since $x_s^T\in \mathrm{ker}(A)$, we have $(Xz_i^T)_s=\sum_{j=1}^nx_{sj}\alpha_{ij}=(Ax_s^T)_i=0$ for all $i=1,\dots,m$ and $s=1,\dots,k$. Thus $U:=\mathrm{span}\set*{z_1^T,\dots,z_m^T}\ssq \mathrm{ker}(X)$.\par 
But we have $\mathrm{dim}(U)=\mathrm{rank}(A)=n-k=\mathrm{dim}(\mathrm{ker}(X))$ and hence $U=\mathrm{ker}(X)$.\par 
Now let 
\begin{equation*}
(\lambda_1,\dots,\lambda_m)^T:=D\log(u/v)^T.
\end{equation*}\par 
Since $\log(u/v)^T\in \mathrm{ker}(X)=U=\mathrm{ran}(A^T)$ and $A^TDA^T=A^T$ it follows that
\begin{equation*}
\sum_{i=1}^m\lambda_iz_i=\log(u/v).
\end{equation*}\par
We put $c_i:=e^{\lambda_i}$ for $i=1,\dots,m$ and $c:=(c_1,\dots,c_m)$. Then, similar to the proof of Theorem \ref{thm:buckingham1}, it follows that $c^{\alpha_j}=u_j/v_j$ for $j=1,\dots,n$.\par 
Let $\gamma=(\gamma_1,\dots,\gamma_m):=yA^T=(Ay^T)^T$. It follows from Lemma \ref{lemma:buckingham} that 
$(v/u)^yc^{\gamma}=\prod_{j=1}^n(c^{\alpha_j}v_j/u_j)^{y_j}=1$.\par
Because of $F\in S(A,\beta,\eps)$ we have
\begin{equation*}
|F(v)c^{\beta}-F(u)|\leq\eps|F(v)|c^{\beta}.
\end{equation*}
and so it follows that
\begin{align}\label{eq:buckingham4}
&|F(v)-G(\pi_1(v),\dots,\pi_k(v))v^y|=|F(v)-G(w)v^y|=|F(v)-v^yF(u)/u^y| \nonumber\\
&=|F(v)-F(u)c^{-\gamma}|\leq|F(v)c^{\beta-\gamma}-F(u)c^{-\gamma}|+|F(v)c^{\beta-\gamma}-F(v)| \nonumber\\
&\leq|F(v)|(\eps c^{\beta-\gamma}+|c^{\beta-\gamma}-1|).
\end{align}\par 
Since $K^{-1}\leq v_j\leq K$ for $j=1,\dots,n$ we have
\begin{equation*}
|\log(v^{x_s})|\leq \sum_{j=1}^n|x_{sj}||\log(v_j)|\leq nM\log(K) \ \ \ \forall s=1,\dots,k.
\end{equation*}\par 
It follows that $\norm{X^{\dagger}(\log(v^{x_1}),\dots,\log(v^{x_k}))^T}_{\infty}\leq nM\norm{X^{\dagger}}\log(K)$.\par  
Now $u=\psi(w)=\psi(v^{x_1},\dots,v^{x_k})=\exp(X^{\dagger}(\log(v^{x_1}),\dots,\log(v^{x_k}))^T)$
implies $K^{-nM\norm{X^{\dagger}}}\leq u_j\leq K^{nM\norm{X^{\dagger}}}$ for $j=1,\dots,n$.\par 
Together with our assumption on $v$ this implies
\begin{equation*}
K^{-nM\norm{X^{\dagger}}-1}\leq \frac{u_j}{v_j}\leq K^{nM\norm{X^{\dagger}}+1} \ \ \ \forall j=1,\dots,n.
\end{equation*}\par
But then $|\lambda_i|\leq \norm{D}\norm{\log(u/v)^T}_{\infty}\leq \norm{D}(nM\norm{X^{\dagger}}+1)\log(K)$ for $i=1,\dots,m$.\par
Since $c_i=e^{\lambda_i}$ we obtain $K^{-\norm{D}(nM\norm{X^{\dagger}}+1)}\leq c_i\leq K^{\norm{D}(nM\norm{X^{\dagger}}+1)}$ for $i=1,\dots,m$.\par
Because of $|\beta_i-\gamma_i|\leq \norm{\beta^T-\gamma^T}_{\infty}\leq\delta$ it follows that $K^{-\delta\norm{D}(nM\norm{X^{\dagger}}+1)}\leq c_i^{\beta_i-\gamma_i}\leq K^{\delta\norm{D}(nM\norm{X^{\dagger}}+1)}$ for $i=1,\dots,m$ and hence 
\begin{equation*}
K^{-m\delta\norm{D}(nM\norm{X^{\dagger}}+1)}\leq c^{\beta-\gamma}\leq K^{m\delta\norm{D}(nM\norm{X^{\dagger}}+1)}.
\end{equation*}\par 
As in the proof of Theorem \ref{thm:buckingham1} this implies $|c^{\beta-\gamma}-1|\leq K^{m\delta\norm{D}(nM\norm{X^{\dagger}}+1)}-1$ and together with \eqref{eq:buckingham4} we obtain 
\begin{equation*}
|F(v)-G(\pi_1(v),\dots,\pi_k(v))v^y|\leq |F(v)|((1+\eps)K^{m\delta\norm{D}(nM\norm{X^{\dagger}}+1)}-1).
\end{equation*} 
\end{proof}

Let us end this paper with a few remarks concerning Theorem \ref{thm:buckingham2}.\par 
\ \\
1) In the special case $\eps=0=\delta$ (the case of the classical Buckingham Theorem)
we have the following uniqueness assertion: the function $G$ defined in Theorem \ref{thm:buckingham2} is the only function on $\R_k^+$ 
satisfying $F(v)=G(\pi_1(v),\dots,\pi_k(v))v^y$ for every $v\in \R_n^+$.\par 
To see this suppose $H:\R_k^+ \rightarrow \R$ is another function such that 
$F(v)=H(\pi_1(v),\dots,\pi_k(v))v^y$ holds for every $v\in \R_n^+$. Fix an arbitrary $w\in \R_k^+$ and put $v:=\psi(w)$. 
Then $X\log(v)^T=\log(w)^T$ and hence $(\pi_1(v),\dots,\pi_k(v))=w$ (see \eqref{eq:buckingham3}). It follows that $H(w)=F(v)/v^y=G(w)$.\par 
\ \\
2) If $Ay^T=\beta^T$ and $F$ is of the form $F(v)=H(\pi_1(v),\dots,\pi_k(v))v^y$ for every $v\in \R_n^+$, where $H:\R_k^+ \rightarrow \R$ 
is an arbitrary function, then we automatically have $F\in S(A,\beta)$ (this follows immediately from Lemma \ref{lemma:buckingham}, 
since $\pi_1,\dots,\pi_k\in S(A,0)$).\par 
\ \\
3) It is worth mentioning that the function $G$ defined in Theorem \ref{thm:buckingham2} is ``as good as $F$ regarding regularity properties'', more precisely:
\begin{enumerate}[(i)]
\item If $F$ is measurable, then $G$ is measurable.
\item If $F$ is continuous, then $G$ is continuous.
\item If $F$ is differentiable, then $G$ is differentiable.
\end{enumerate}
(This follows immediately from the usual rules.)\\
\ \\
4) For every real $m\times n$-matrix $A\neq 0$, every $\eps>0$ and every $\beta^T$ in the range of $A$, we have that $S(A,\beta)$ is a proper subset of $S(A,\beta,\eps)$.

\begin{proof}
Let $y\in \R_n$ such that $Ay^T=\beta^T$ and fix $\tau\in (0,1)$ such that
\begin{equation*}
\frac{1+\tau}{1-\tau}\leq 1+\eps \ \ \ \text{and}\ \ \ \frac{1-\tau}{1+\tau}\geq 1-\eps.
\end{equation*}
\indent Now we first suppose that $H:\R_n^+ \rightarrow [1-\tau,1+\tau]$ is an arbitrary function and 
we put $F(v):=H(v)v^y$ for every $v\in \R_n^+$. Then $F\in S(A,\beta,\eps)$.\par 
To see this we denote the columns of $A$ by $\alpha_1^T,\dots,\alpha_n^T$ and let $v=(v_1,\dots,v_n)\in \R_n^+$ and $c\in \R_m^+$. Then we have, according to Lemma \ref{lemma:buckingham}, 
$\prod_{j=1}^n(v_jc^{\alpha_j})^{y_j}=v^yc^{\beta}$ and hence
\begin{align*}
&|F(v_1c^{\alpha_1},\dots,v_nc^{\alpha_n})-F(v_1,\dots,v_n)c^{\beta}|\\
&=|H(v_1c^{\alpha_1},\dots,v_nc^{\alpha_n})-H(v_1,\dots,v_n)|v^yc^{\beta}\\
&\leq v^yc^{\beta}\eps|H(v_1,\dots,v_n)|=\eps|F(v_1,\dots,v_n)|c^{\beta},
\end{align*}
where the ``$\leq$'' holds because of
\begin{equation*}
\abs*{1-\frac{H(v_1c^{\alpha_1},\dots,v_nc^{\alpha_n})}{H(v_1,\dots,v_n)}}\leq
\max\set*{\frac{1+\tau}{1-\tau}-1,1-\frac{1-\tau}{1+\tau}}\leq\eps.
\end{equation*}
\indent Now we define a specific function $H:\R_n^+ \rightarrow [1-\tau,1+\tau]$ by $H(1,\dots,1):=1-\tau$ and $H(v):=1+\tau$ for every $v\in \R_n^+\sm \set*{(1,\dots,1)}$.\par 
The corresponding $F$ is in $S(A,\beta,\eps)$ and we claim that $F\not\in S(A,\beta)$.\par 
To see this we distinguish two cases. If $\mathrm{rank}(A)=n$, then $F\in S(A,\beta)$ would
imply that $H$ is constant (Theorem \ref{thm:buckingham1}), which is not true.\par
If $\mathrm{rank}(A)<n$ and $F\in S(A,\beta)$, then, in the notation of Theorem \ref{thm:buckingham2}, we would have $H(v)=G(\pi_1(v),\dots,\pi_k(v))$ for every $v\in \R_n^+$.\par 
Since $\pi_1,\dots,\pi_k\in S(A,0)$ this would mean $H\in S(A,0)$. But $A\neq 0$, so we can always find $c\in \R_m^+$ such $(c^{\alpha_1},\dots,c^{\alpha_n})\neq (1,\dots,1)$ and hence 
$H(c^{\alpha_1},\dots,c^{\alpha_n})\neq H(1,\dots,1)$, thus $H\not\in S(A,0)$.\par 
This completes the proof.
\end{proof}

\newpage

\section{Appendix}\label{sec:appendix}
\ \\
{	
	\renewcommand{\arraystretch}{1.3}
\begin{center}
{\bf Table 2: Examples of derived quantities and their  SI units}\\
\ \\
\begin{tabular}{|c|c|}
\hline
{\bf Quantity} \ \ & {\bf Unit} \\ \hline
area \ \ & $\si{m^2}$ \\ \hline
volume \ \ & $\si{m^3}$ \\ \hline
density \ \ & $\si{m^{-3}.kg}$ \\ \hline
velocity \ \ & $\si{m.s^{-1}}$ \\ \hline
acceleration \ \ & $\si{m.s^{-2}}$ \\ \hline
momentum \ \ & $\si{m.kg.s^{-1}}$ \\ \hline 
angular momentum \ \ & $\si{m^2.kg.s^{-1}}$ \\ \hline
moment of inertia \ \ & $\si{m^2.kg}$ \\ \hline
force \ \ & $\si{m.kg.s^{-2}}=\si{N}$ (newton) \\ \hline 
energy \ \ & $\si{m^2.kg.s^{-2}}=\si{J}$ (joule) \\ \hline
frequency \ \ & $\si{s^{-1}}=\si{Hz}$ (hertz) \\ \hline 
pressure \ \ & $\si{m^{-1}.kg.s^{-2}}=\si{N/m^2}=\si{Pa}$ (pascal) \\ \hline
power \ \ & $\si{m^2.kg.s^{-3}}=\si{J/s}=\si{W}$ (watt) \\ \hline
electric charge \ \ & $\si{s.A}=\si{C}$ (coulomb) \\ \hline 
electric voltage \ \ & $\si{m^2.kg.s^{-3}.A^{-1}}=\si{J/C}=\si{V}$ (volt) \\ \hline
electrical resistance \ \ & $\si{m^2.kg.s^{-3}.A^{-2}}=\si{V/A}=\si{\Omega}$ (ohm) \\ \hline
electrical conductance \ \ & $\si{m^{-2}.kg^{-1}.s^3A^2}=\si{\Omega^{-1}}=\si{S}$ (siemens) \\ \hline
capacitance \ \ & $\si{m^{-2}.kg^{-1}.s^4.A^2}=\si{C/V}=\si{F}$ (farad) \\ \hline 
magnetic flux \ \ & $\si{m^2.kg.s^{-2}.A^{-1}}=\si{Wb}$ (weber) \\ \hline
magnetic flux density \ \ & $\si{kg.s^{-2}.A^{-1}}=\si{Wb/m^2}=\si{T}$ (tesla) \\ \hline
inductance \ \ & $\si{m^2.kg.s^{-2}.A^{-2}}=\si{Wb/A}=\si{H}$ (henry) \\ \hline
entropy \ \ & $\si{m^2.kg.s^{-2}.K^{-1}}=\si{J/K}$ \\ \hline
radioactivity \ \ & $\si{s^{-1}}=\si{Bq}$ (becquerel) \\ \hline
ionizing radiation dose \ \ & $\si{m^2.s^{-2}}=\si{J/kg}=\si{Gy}$ (gray) \\ \hline
equivalent dose \ \ & $\si{m^2.s^{-2}}=\si{J/kg}=\si{Sv}$ (sievert) \\ \hline
catalytic activity \ \ & $\si{s^{-1}.mol}=\si{kat}$ (katal) \\ \hline
\end{tabular}
\end{center}
	}

\newpage

{	
	\renewcommand{\arraystretch}{1.3}
\begin{center}
{\bf Table 3: SI prefixes}	
\end{center}

\begin{tabular}{|c|c|}
\hline
{\bf Prefix} \ \ & {\bf Scaling factor} \\ \hline
yotta (\si{Y}) \ \ & $10^{24}$ \\ \hline
zetta (\si{Z}) \ \ & $10^{21}$ \\ \hline
exa (\si{E}) \ \ & $10^{18}$ \\ \hline
peta (\si{P}) \ \ & $10^{15}$ \\ \hline
tera (\si{T}) \ \ & $10^{12}$ \\ \hline
giga (\si{G}) \ \ & $10^9$ \\ \hline
mega (\si{M}) \ \ & $10^6$ \\ \hline
kilo (\si{k}) \ \ & $10^3$ \\ \hline
hecto (\si{h}) \ \ & $10^2$ \\ \hline
deca (\si{da}) \ \ & $10$ \\ \hline
\end{tabular}
\ \ \ \ \ \ \ \ 
\begin{tabular}{|c|c|}
\hline
{\bf Prefix} \ \ & {\bf Scaling factor} \\ \hline
deci (\si{d}) \ \ & $10^{-1}$ \\ \hline
centi (\si{c}) \ \ & $10^{-2}$ \\ \hline
milli (\si{m}) \ \ & $10^{-3}$ \\ \hline
micro ($\si{\mu}$) \ \ & $10^{-6}$ \\ \hline
nano (\si{n}) \ \ & $10^{-9}$ \\ \hline
pico (\si{p}) \ \ & $10^{-12}$ \\ \hline
femto (\si{f}) \ \ & $10^{-15}$ \\  \hline
atto (\si{a}) \ \ & $10^{-18}$ \\ \hline
zepto (\si{z}) \ \ & $10^{-21}$ \\ \hline
yocto (\si{y}) \ \ & $10^{-24}$ \\ \hline
\end{tabular}
	
}

{	
	\renewcommand{\arraystretch}{1.3}
\begin{center}
{\bf Table 4: Examples of non-SI units}\\
\ \\
\begin{savenotes}
\begin{tabular}{|c|c|c|}
\hline
{\bf Quantity} \ \ & {\bf Unit} \ \ & {\bf Value\footnote{Factors are rounded to five decimal digits.}} \\ \hline
\multirow{4}{*}{time}
& minute (\si{min}) \ \ & $60\,\si{s}$ \\ \cline{2-3}
& hour (\si{h}) \ \ & $60\,\si{min}=3600\,\si{s}$ \\ \cline{2-3}
& day (\si{d}) \ \ & $24\,\si{h}=86400\,\si{s}$ \\ \cline{2-3}
& year (\si{a}) \ \ & $365.25\,\si{d}$ \\ \hline
\multirow{4}{*}{length}
& \r{a}ngstr\"om (\si{\r{A}}) \ \ & $10^{-10}\,\si{m}$ \\ \cline{2-3}
& astronomical unit (\si{au}) \ \ & $1.49598\cdot 10^{11}\,\si{m}$ \\ \cline{2-3}
& light-year (\si{ly}) \ \ & $9.46073\cdot10^{15}\,\si{m}$ \\ \cline{2-3}
& parsec (\si{pc}) \ \ & $3.08568\cdot10^{16}\,\si{m}\approx 3.26\,\si{ly}$ \\ \hline
\multirow{2}{*}{mass}
& unified atomic mass unit (\si{u}) \ \ & $1.66054\cdot10^{-27}\,\si{kg}$ \\ \cline{2-3} 
& tonne (\si{t}) \ \ & $10^3\,\si{kg}$ \\ \hline
\multirow{2}{*}{area}
& barn (\si{b}) \ \ & $10^{-28}\,\si{m^2}$ \\ \cline{2-3}
& hectare (\si{ha}) \ \ & $10^4\,\si{m^2}$ \\ \hline
volume \ \ & liter (\si{l}) \ \ & $10^{-3}\,\si{m^3}$ \\ \hline
\multirow{2}{*}{energy}
& electronvolt (\si{eV}) \ \ & $1.60218\cdot10^{-19}\,\si{J}$ \\ \cline{2-3}
& calorie (\si{cal}) \ \ & $4.184\,\si{J}$ \\ \hline
\multirow{2}{*}{pressure}
& bar (\si{bar}) \ \ & $10^5\,\si{Pa}$ \\ \cline{2-3}
& standard atmosphere (\si{atm}) \ \ & $101325\,\si{Pa}$ \\ \hline
\end{tabular}
\end{savenotes}
\end{center}
    }

\ \\
\address
\email 

\end{document}